\documentclass{amsart}
\usepackage[margin=1in]{geometry}
\usepackage{algorithm, algorithmic}
\usepackage{amsthm, amssymb, amsmath, color, amsfonts}
\usepackage{url}
\usepackage{graphicx, subfigure} 
\usepackage{xspace}

\renewcommand{\AA}{\mathcal{A}}
\newcommand{\R}{\mathcal{R}}
\newcommand{\RR}{\mathbb{R}}

\newcommand{\Z}{\mathcal{Z}}
\newcommand{\C}{\mathcal{C}}

\newcommand{\axiom}[1]{\texttt{#1}\xspace}
\newcommand{\ADD}{\axiom{Additivity}}
\newcommand{\DUM}{\axiom{Dummy}}
\newcommand{\DUMP}{\axiom{Dummy'}}

\newcommand{\SI}{\axiom{Scale Invariance}}
\newcommand{\ASI}{\axiom{Affine Scale Invariance}}
\newcommand{\MON}{\axiom{Monotonicity}}
\newcommand{\CON}{\axiom{Conditional Nonnegativity}}
\newcommand{\ANON}{\axiom{Anonymity}}
\DeclareMathOperator{\spann}{span}

\theoremstyle{definition}
\numberwithin{equation}{section}
\newtheorem{theorem}{Theorem}[section]
\newtheorem{define}[theorem]{Definition}
\newtheorem{lemma}[theorem]{Lemma}
\newtheorem{cor}[theorem]{Corollary}
\newtheorem{example}[theorem]{Example}
\newtheorem*{remark}{Remark}

\begin{document}

\title{Axiomatic Attribution for Multilinear Functions}

\author{Yi Sun}
\address{Churchill College, Cambridge CB3 0DS, United Kingdom}
\email{yisun@mit.edu}

\author{Mukund Sundararajan}
\address{Google Research\\ 1600 Amphitheatre Parkway\\Mountain View, CA 94043}
\email{mukunds@google.com}

\date{\today}

\thanks{The authors would like to thank the anonymous referees for helpful comments which improved the exposition and in particular for pointing out the reference \cite{Owen}.}

\begin{abstract}
We study the attribution problem, that is, the problem of attributing a change in the value of a  characteristic function $f$ to its independent variables.  We make three contributions.
First, we propose a formalization of the problem based on a standard cost sharing model.
Second, we show that there is a unique attribution method that satisfies \DUM, \ADD, \CON, \ASI, and \ANON for all characteristic functions that are the sum of a multilinear function and an additive function. We term this the \textit{Aumann-Shapley-Shubik} method. Conversely, we show that such a uniqueness result does not hold for characteristic functions outside this class. Third, we study multilinear characteristic functions in detail; we describe a computationally efficient implementation of the Aumann-Shapley-Shubik method and discuss practical applications to pay-per-click advertising and portfolio analysis.
\end{abstract}

\maketitle

\section{Introduction}
\subsection{The Attribution Problem} \label{theproblem}

Consider a function $f(r_1, \ldots, r_n)$ of several variables $r_1, \ldots, r_n$.  Given a change in the values of these variables, we ask what portion of the overall change is due to the change in each variable $r_i$. In particular, we would like to divide the responsibility for the overall change among the variables in an axiomatic way.  We term such problems \emph{attribution problems} and the responsibilities assigned \emph{attributions}.  The attribution to the $i^\text{th}$ variable can be more interesting than simply the change $s_i - r_i$ in the variable because the relationship between the magnitude of the change in a variable and the impact it has on $f$ depends on the form of $f$. For instance, a tiny change in a variable could have a huge impact on the the value of the function.
 
Formally, we are given a real-valued characteristic function $f: \RR^n \to \RR$ of $n$ variables and initial and final values $r_i$ and $s_i$ for the independent variables.  Here, the function $f$ is deterministic, not learned from data, and the values of $r$ and $s$ are known exactly and are not estimates in any sense. Our objective is to find attributions $z_1(r, s, f), \ldots, z_n(r, s, f)$, where we interpret $z_i(r, s, f)$ as the portion of the change in $f$ due to the change in the $i^\text{th}$ variable, so that $z_1(r, s, f) + \cdots + z_n(r, s, f) = f(s) - f(r)$, which we call the \emph{completeness} condition on the attribution.  We interpret completeness as meaning that all the change in $f$ is accounted for. (We often omit the characteristic function and simply write $z_i(r, s)$ for $z_i(r, s, f)$.)

As we discuss attribution, we will keep the following motivating example in mind. See Section~\ref{applications} for other examples, and a broader discussion about the applicability of our techniques.
\begin{example}\label{ex:simple}
Consider a firm that repeatedly procures a good from a foreign supplier for use in its manufacturing process.  It incurs some expenditure, the product $e = a \cdot p \cdot c$ of the amount $a$ of the good that the buyer purchases, the average cost per unit $p$ of the good in the foreign currency, and the conversion rate $c$ between the foreign and local currencies.  We take $e(a, p, c)$ as the characteristic function.  The final values of $e$, $a$, $p$, and $c$ may be statistics from a certain quarter, and the initial values may be statistics from the preceding quarter.  The attribution problem, then, is to divide responsibility for the change in $e$ among the changes in $a$, $p$, and $c$.

Suppose further that the demand for the good comes from the manufacturing department (so an improvement in the efficiency of manufacturing reduces $a$), that the price for the good is negotiated by the procurement department (so an improvement in the negotiation process decreases $p$), and that the exchange rate is exogenously determined.  Such an attribution could then serve to apportion blame between or determine bonuses for the two departments.
\end{example}

How can we attribute the change in the characteristic function $f$ to the various variables? If $f$ were linear, that is, if $f$ takes the form $f(r_1, \ldots, r_n) = \sum_{i}b_ir_i$, then for a change from $r$ to $s$, it is natural to attribute $b_i \cdot (s_i - r_i)$ to the $i^\text{th}$ variable. For non-linear functions such as the one in the Example~\ref{ex:simple}, if the independent variables all change slightly, we could replace $b_i$ by the partial derivative with respect to $r_i$ at $s$, giving a linear approximation of $f$ locally at the final value and performing attribution as above. 

However, if the changes in the variables are not slight, then this approach would badly violate completeness. For instance, in Example~\ref{ex:simple}, suppose $a$ changes from $4$ to $5$, $p$ changes from $1$ to $12$, and $c$ changes from $1$ to $1.5$. Using this approach, the attributions to $a$, $p$, and $c$ are $(5 - 4) \cdot 12 \cdot 1.5 = 18$, $5 \cdot (12 - 1) \cdot 1.5 = 82.5$, and $5 \cdot 12 \cdot (1.5 - 1) = 30$, respectively.  This assigns the two departments and the exogenous currency rate change blame for a total of $18 + 82.5 + 30 = 130.5$ of change; but the total change is only $5 \cdot 12 \cdot 1.5 - 4 \cdot 1 \cdot 1 = 86$.  This means that the attribution violates completeness, making it difficult to interpret practically.  

It might seem in this example that the failure of completeness originated from a poor choice of point approximation for the partial derivative of $f$.  In general, no systematic use of such a point approximation suffices for our application.  However, in Subsection~\ref{sec-path}, we will examine a principled method of computing attributions along these lines.

\subsection{Axioms for attribution methods} \label{axioms}

Our attribution problem (almost trivially) generalizes the cost or surplus sharing problem from the social choice literature (cf. Moulin~\cite{Mou}), where the problem is to axiomatically share the cost of production or surplus among several agents.\footnote{This is also sometimes called the fair division problem.} The characteristic function is either cost or surplus, the independent variables correspond to demands or contributions of agents, and the attributions correspond to cost shares or profit shares. The completeness condition is the \emph{budget balance} condition for cost sharing.  We give a more detailed discussion of the relationship between these two problems in Subsection~\ref{versus}.

Following the cost sharing literature, we take an axiomatic approach to choosing methods to use for attribution problems. In this section we discuss the axioms we consider and briefly discuss motivations for them, emphasizing the attribution context; see the cited papers for a longer discussion in the cost sharing context. 

\begin{itemize}
\item \DUM: If the value of the characteristic function does not depend on a variable, then the attribution to that variable is zero. 

\noindent This axiom is very natural, as it simply requires that variables irrelevant to the outcome be ignored.\footnote{In the cost sharing context, \DUM is the bedrock of the no cross-subsidy (full-responsibility) theory. In such a theory, the variable is deemed responsible for asymmetries in the cost as well as asymmetries in the demand; see Moulin and Sprumont~\cite{MS} for a discussion.}

\item \DUMP: If the value of the characteristic function $f$ does not depend on a variable $r_i$ on $[r, s]$, then the attribution $z_i(r, s, f)$ to that variable is zero.

\noindent This axiom is a natural strengthening of \DUM.  It may be viewed as a local version of the global axiom \DUM.

\item \ADD: For all $r, s, f_1, f_2$, we have that $z_i(r, s, f_1 + f_2) = z_i(r, s, f_1) + z_i(r, s, f_2)$.  

\noindent This axiom yields a type of procedural invariance.  That is, if the system modeled by the characteristic function can be decomposed into several independent sub-processes that interact additively, we can compute the attributions separately for each sub-process. Alternatively, \ADD can be justified via \emph{lex parsimonae}. Constructing attributions is equivalent to linearizing the effect of changes in the independent variables. When an attribution method satisfies \ADD, it is minimal in the sense that it preserves the pre-existing linear structure of the characteristic function.

\item \ANON: The attributions are unchanged by relabeling of the variables. More formally, for any permutation $\sigma \in S_n$, if $f_\sigma(r_1, \ldots, r_n) = f(r_{\sigma^{-1}(1)}, \ldots, r_{\sigma^{-1}(n)})$, then for all $i$, we have
\[
z_{\sigma^{-1}(i)}(r_{\sigma(1)}, \ldots, r_{\sigma(n)}, s_{\sigma(1)}, \ldots, s_{\sigma(n)}, f_\sigma) = z_i(r, s, f).
\]

\noindent \ANON conveys the idea that all variables in the characteristic function should be treated equally, up to their initial and final values.

\item \CON: Suppose the characteristic function $f$ is non-decreasing in a variable $i$ on $[r, s]$. Then for all $r, s$, if $s_i \geq r_i$ (resp. $s_i \leq r_i$), then $z_i(r, s ,f) \geq 0$ (resp. $z_i(r, s, f) \leq 0$).

\item \MON~\cite{FM}: Suppose the characteristic function $f$ is non-decreasing in variable $j$.  Then, for input pairs $(r, s)$ and $(r, s')$ such that $s_i = s_i'$ for $i \neq j$ and $s_j < s_j'$, we have $z_j(r, s, f) \leq z_j(r, s', f)$.

\noindent \MON, and \CON preclude attributions with counterintuitive signs.

\item \SI~\cite{FM}: The attributions are independent of linear rescaling of individual variables.  That is, for any $c > 0$, if $g(r_1, \ldots, r_n) = f(r_1, \ldots, r_j/c, \ldots, r_n)$, then for all $i$ we have
\[
z_i(r, s, f) = z_i\Big((r_1, \ldots, cr_j, \ldots, r_n), (s_1, \ldots, cs_j, \ldots, s_n), g\Big).
\]

\noindent \SI conveys the idea that the attributions should be independent of the (possibly incomparable) units in which individual variables are measured.  It is especially compelling in the context of attribution because the different variables may refer to quantities of entirely different things.

\item \ASI~\cite{SW}: The attributions are invariant under simultaneous affine transformation of the characteristic function and the variables. That is, for any $c, d > 0$, if $g(r_1, \ldots, r_n) = f(r_1, \ldots, (r_j - d)/c, \ldots, r_n)$, then for all $i$ we have
\[
z_i(r, s, f) = z_i\Big((r_1, \ldots, cr_j + d, \ldots, r_n), (s_1, \ldots, cs_j + d, \ldots, s_n), g\Big).
\]

\noindent \ASI conveys the idea that both the units and the zero point of individual variables should not affect the value of the attribution.  Again, for attribution this is especially compelling, since the variables may represent values without naturally defined units or zero points.  For example, temperature is commonly measured in both Celsius and Fahrenheit scales, which are related by an affine transformation.

\end{itemize}

We include the \SI and \MON axioms only to facilitate discussion and comparison with axiomatizations of attribution methods in the cost sharing literature. They will play no role in the main results.

\subsection{Candidate attribution methods} 

In this section we describe some attribution methods motivated by the cost sharing literature, and mention the axioms they satisfy.

\subsubsection{Path methods} \label{sec-path}

We first consider a natural class of attribution methods, the path attribution methods, that are well-studied in the cost sharing context (see \cite{FM,Hai}).  These methods assign to each variable its marginal effect along some path from the initial point to the final point.  They are analogous to the approach based on partial derivatives outlined in Subsection \ref{theproblem}, but they salvage completeness by integrating the partial derivatives along a path instead of taking a naive estimate at a single endpoint.   Note that their definition is motivated by Theorem~\ref{representation} from the cost sharing literature, which we discuss in Subsection~\ref{versus}.

\begin{define} \label{def:path-attribution}
For each $r, s \in \RR^n$, let $\gamma_{r, s}: [0, 1] \to \RR^n$ be a $C^1$-function with $\gamma_{r, s}(0) = r$ and $\gamma_{r, s}(1) = s$, which we interpret as a path from $r$ to $s$.  Write $\gamma_{r, s} = (\gamma_{r, s, 1}, \ldots, \gamma_{r, s, n})$, and let $\gamma_{r, s, i}$ be non-decreasing if $r_i \leq s_i$ and non-increasing if $r_i \geq s_i$.  Then, the attribution method given by
\begin{equation} \label{speq}
z_i(r, s) = \int_0^1 \partial_i f(\gamma_{r, s}(t)) \gamma_{r, s, i}'(t) dt
\end{equation}
is the \textit{single-path attribution method} corresponding to the family of paths $\gamma_{r, s}$.  If the method
\[
z_i(r, s) = \sum_j c_j z^j_i(r, s) \text{ for $c_j \geq 0$ and $\sum_j c_j = 1$}
\]
is a convex combination of single-path attribution methods $z^j$, we say that $z$ is a \textit{path attribution method}.
\end{define}

For a single-path attribution method, we may check by the gradient theorem that 
\[
z_1(r, s) + \cdots + z_n(r, s) = \int_0^1 \sum_{i = 1}^n \partial_i f(\gamma_{r, s}(t)) \gamma_{r, s, i}'(t) dt = \int_{\gamma_{r, s}} \nabla \cdot f = f(s) - f(r),
\]
meaning that completeness is satisfied for each single-path attribution method.  Completeness is preserved under convex combinations and therefore holds for all path attribution methods.  Further, path attribution methods satisfy \DUM, \DUMP, \ADD, and \CON for all characteristic functions.  That \DUM and \DUMP hold is obvious, \ADD holds because partial differentiation is linear, and \CON holds because a characteristic function non-decreasing in a variable has a non-negative partial derivative with respect to that variable. 

In the cost sharing context, Theorem~\ref{representation} implies that these are essentially the only methods that satisfy \ADD and \DUM for all characteristic functions.  We suspect that an analogue of this result also holds in the attribution context, which guides our intuition. While none of our formal results rely on the notion of path methods, they provide a convenient way to think about attribution methods and are a useful starting point in our investigation of desirable attribution methods.  We will now identify some specific candidate path attribution methods, which will use the following general construction.

\begin{define} \label{scaled-path-attribution}
Fix a path $\gamma: [0, 1] \to [0, 1]^n$, non-decreasing in each variable, such that $\gamma(0) = 0$ and $\gamma(1) = (1, \ldots, 1)$.  Write $\gamma = (\gamma_1, \ldots, \gamma_n)$.  Then, the single-path attribution method corresponding to 
\[
\gamma_{r, s}(t) = r + \Big((s_1 - r_1) \gamma_1(t), \ldots, (s_n - r_n) \gamma_n(t)\Big)
\]
is the \textit{affine single-path attribution method} corresponding to $\gamma$.  An \textit{affine path attribution method} is a convex combination of affine single-path attribution methods.
\end{define}

\subsubsection{Methods based on the Shapley value} \label{shapley-based}

Recall that the Shapley value (Shapley~\cite{Shapley}) is a solution concept in cooperative game theory used to distribute the total surplus generated by a coalition of players among the players. We now mention two attribution methods that are adaptations of this discrete solution concept; both methods have been well-studied in the context of continuous demand cost sharing. We start with the method that is arguably the best known method in the cost sharing literature.

\begin{define} \label{AS}
The \textit{Aumann-Shapley} method~\cite{AS} is the affine single-path attribution method corresponding to the path $\gamma_i(t) = t$. 
\end{define}

This was identified by Aumann and Shapley~\cite{AS} as a `value' for non-atomic games. Next, we define a different and arguably more direct generalization of the Shapley value.

\begin{define} \label{SS}
 The \textit{Shapley-Shubik} method~\cite{FM,SS} is  defined as follows. For any $\sigma \in S_n$, let $\gamma^\sigma$ be the path
\[
\gamma_i^\sigma(t) = \begin{cases} 0 & tn < \sigma(i)-1 \\ (tn - \sigma(i)) & \sigma(i)-1 \leq tn < \sigma(i) \\ 1 & tn \geq \sigma(i), \end{cases}
\]
where $\gamma^\sigma$ walks along edges of the hypercube $[0, 1]^n$ in an order determined by $\sigma$.  Then, the Shapley-Shubik method is given by the average of the $n!$ path attribution methods corresponding to $\gamma^\sigma$. 

More generally, a \textit{random order method}~\cite{FM, Weber} is any convex combination of the affine path attribution methods corresponding to the $\gamma^\sigma$.  A \textit{value-variant random order method} is a path attribution method such that every path in the corresponding families of paths takes the form $\gamma_{r, s}^\sigma$.
\end{define}
\begin{remark}
Value-variant random order methods refine the notion of random order method in the following sense.  If we fix the initial and final values $r$ and $s$, the attributions $z_i(r, s)$ of a value-variant random order method are a convex combination of the values given by applying (\ref{speq}) for the paths $\gamma_{r, s}^\sigma$.  Such a method is a random order method if the weights of this convex combination do not depend on $r$ and $s$.
\end{remark}

The Aumann-Shapley method and value-variant random order methods (hence random order methods and the Shapley-Shubik method) satisfy \ADD, \DUM, and \DUMP because they are path attribution methods. The Aumann-Shapley method and random order methods (and hence the Shapley-Shubik method) satisfy \ASI by Lemma \ref{affineasi} because they are affine path attribution methods. The Aumann-Shapley and Shapley-Shubik methods additionally satisfy \ANON.

\begin{remark}
We may relate the Shapley-Shubik method to the Aumann-Shapley method as follows. The Shapley-Shubik attribution for a change from $r$ to $s$ is the expected attribution of a monotone random walk along the edges of the hypercube with opposite vertices at $r$ and $s$. If we subdivide the hypercube with opposite vertices at $r$ and $s$ into a grid of smaller hypercubes and consider monotonic random walks in this structure, the density of the resulting walks will be focused on the diagonal.  Hence, when the characteristic function satisfies some basic regularity conditions, the average of the path attribution methods corresponding to these walks will tend to the Aumann-Shapley method in the limit.
\end{remark}

\subsection{Statement of results} \label{sec-results}

When choosing an attribution method, it is very desirable to have a uniqueness result, one which says that there is exactly one method satisfying some axioms because such a result identifies a method for use.  If an attribution method is the unique method satisfying some axioms on a class of characteristic functions, we term these axioms an \textit{axiomatization} for the attribution method. We seek axiomatizations for a specific class of characteristic functions, namely those that are a sum of a additively separable function and a multilinear function, defined as follows. 

\begin{define} \label{functype}
A function $f: \RR^n \to \RR$ is \textit{additively separable} if there exist $f_i: \RR \to \RR$ with
\[
f(r_1, \ldots, r_n) = f_1(r_1) + \cdots + f_n(r_n).
\]
A function $f: \RR^n \to \RR$ is \textit{multilinear} if we may write $f$ in the form
\[
f(r_1, \ldots, r_n) = \sum_{I \subset [n]} c_I \prod_{i \in I} r_i,
\]
that is, as the sum of monomials of degree at most $1$ in each variable.
\end{define}

We justify our focus on a narrow class of characteristic functions in Section~\ref{sec-comparison}, and we demonstrate that such characteristic functions have several practical applications in Section~\ref{applications}.  We defer these considerations for now to state our results.  

Our main result is that there is a unique attribution method that satisfies \DUM, \ADD, \ANON, \CON and \ASI for all characteristic functions that are the sum of a multilinear function and an additive function (Theorem~\ref{unique-ASI-Anon}). Interestingly, Theorem~\ref{AS-SS-agree} shows that the Aumann-Shapley (Definition~\ref{AS}) and Shapley-Shubik (Definition~\ref{SS}) methods, both of which satisfy the axioms mentioned above, coincide for these characteristic functions. We therefore term this the \textit{Aumann-Shapley-Shubik} method. We give an efficient algorithm to compute it in Theorem~\ref{compute} and Corollary~\ref{compute-iterate}.

As an intermediate step toward proving Theorem~\ref{unique-ASI-Anon}, we show that the only methods that satisfy \ADD, \DUM' and \CON for all multilinear characteristic functions are value-variant random order methods (Theorem~\ref{pointwise}). Surprisingly, the proof implies that every path method (a continuous concept) is equivalent to some value-variant random order method (a combinatorial concept) for a multilinear characteristic function. 

To complete our results, we show in Theorem~\ref{converse} that for every characteristic function outside this class, no analog of Theorem~\ref{unique-ASI-Anon} is possible.  That is, we show that the Aumann-Shapley and the Shapley-Shubik methods coincide if and only if the characteristic function is the sum of a multilinear and an additively separable characteristic function. This shows that our restriction to this class of characteristic functions is not simply a technical convenience and provides in Corollary \ref{axiom-ASS} an axiomatization of the Aumann-Shapley-Shubik method. Section~\ref{sec-comparison} discusses further implications of this result.

\subsection{Attribution versus cost sharing} \label{versus}

In this subsection, we discuss the relationship between our attribution problem and the classical cost sharing problem. 

\subsubsection{Cost sharing as attribution}

Cost sharing models come in various flavors depending on whether the demands are binary, integral, or real-valued and whether the cost function is homogeneous or not (see Moulin~\cite{Mou} for a classification). In this sense our model resembles the $rr$, heterogeneous cost sharing model (both the characteristic function and the independent variables are real-valued, and the characteristic function is not homogeneous in the variables). More precisely, for a monotonically increasing cost function, $rr$ heterogeneous cost sharing is equivalent to attribution from $0$ to the final demand. 

There are two immediate differences between attribution and cost sharing.  First, in the attribution problem, variables change from one set of values to another, while in cost sharing there is just a single set of demands or contributions.  Secondly, attribution relaxes the requirement that the cost function be monotone. Therefore, while negative cost shares do not make sense, negative attributions can make sense in some contexts.

\begin{remark}
A naive approach to attribution might be to determine the attributions $z_i(r, s, f)$ as the difference of the cost sharing problems $z_i(0, s, f) - z_i(0, r, f)$.  Given a valid cost sharing method, this defines a valid attribution method.  However, this approach would not reflect the behavior of the characteristic function $f$ between $r$ and $s$, instead expressing the idea that the change in $f$ is from $r$ to $0$ and then from $0$ to $s$. More formally, \DUM' and \ASI are not satisfied by this approach. This suggests that naively applying the cost sharing framework may not be appropriate in this case.
\end{remark}

\subsubsection{Axiomatics for cost sharing}

Here we discuss axiomatization results from the cost sharing literature, both as motivation for some of our assumptions and as context for our results.   We begin with a result identifying the analogue of the path attribution methods in cost sharing as exactly those methods satisfying the most basic of the axioms we introduced in Subsection~\ref{axioms}. 

\begin{theorem}[Theorem 1 of \cite{Fri}] \label{representation}
Any cost sharing method satisfying \DUM and \ADD is a path cost sharing method.  
\end{theorem}

We now give two axiomatizations of the Aumann-Shapley and Shapley-Shubik methods in the cost sharing context.  The Aumann-Shapley method was axiomatized by Billera and Heath~\cite{BH} and Mirman~\cite{Mir} in the following theorem. 

\begin{theorem}[\cite{BH, FM, Mir}] \label{AS-unique}
The Aumann-Shapley method is the unique cost sharing method that satisfies \ADD, \DUM, \SI, and \texttt{Average Cost for Homogeneous Goods}, which states that, for cost functions that are a function of the sum of the demands, the cost shares should be proportional to the demands.
\end{theorem}

For the Shapley-Shubik method, we have the following axiomatization given by Friedman and Moulin~\cite{FM}.

\begin{theorem}[Theorem 1 of \cite{FM}] \label{thm:SS}
Any cost sharing method satisfying \ADD, \DUM, \MON, \SI, and \axiom{Continuity at Zero} (cost shares are continuous in each variable near $0$) is a random order method.  The Shapley-Shubik method is the unique cost sharing method that satisfies \ANON in addition to \ADD, \DUM, \MON, \SI, and \axiom{Continuity at Zero}.
\end{theorem}

\begin{remark}
In the attribution context, the Aumann-Shapley and Shapley-Shubik methods satisfy the axioms of Theorems~\ref{AS-unique} and~\ref{thm:SS}, but it is not clear if the uniqueness properties continue to hold.  We suspect that these results should also carry over to the attribution framework (with very similar proofs) after some appropriate modification.
\end{remark}

\subsubsection{Axiomatization for attribution versus axiomatization for cost sharing} \label{sec-comparison}

Our approach to the axiomatic study of attribution methods differs from that taken in the cost sharing literature.  A typical axiomatic result in the cost sharing literature (like Theorems~\ref{AS-unique} and~\ref{thm:SS}) identifies a certain cost sharing method as the unique method that satisfies certain axioms for \emph{all} cost functions (cf.~\cite{FM,Shapley}). This does not preclude the existence of multiple methods that satisfy the same set of axioms for a certain subclass of cost functions. For instance, Redekop~\cite{Red} notes that the Aumann-Shapley cost sharing method satisfies the axioms mentioned in the uniqueness result for the Shapley-Shubik method (Theorem~\ref{thm:SS}) when the cost function has increasing marginal costs (i.e. when the cost function is convex).\footnote{It is easy to show directly that all path attribution methods satisfying \SI are monotone for convex functions, of which multilinear functions with positive coefficients are an instance. Redekop~\cite{Red} notices this for the Aumann-Shapley attribution method.} 

In our model, the characteristic function is known when the attribution method is selected, so general uniqueness results similar to Theorems~\ref{AS-unique} and~\ref{thm:SS} are not necessarily sufficient to guide the selection of an attribution method. For instance, for a specific convex characteristic function, there might be more than one attribution method which satisfies the axioms required in Theorems \ref{AS-unique} and \ref{thm:SS}, meaning that they are not enough to select a unique method; in fact, all the applications in Section~\ref{applications} have convex characteristic functions. We therefore seek and successfully identify (see Section~\ref{sec-results}) axiomatizations that quantify less universally over the space of characteristic functions.

In addition, quantifying less universally over the space of characteristic functions allows us to be more parsimonious with axioms.  In the case of multilinear functions, our main result allows us to characterize the Aumann-Shapley method without using \axiom{Average Cost for Homogeneous Goods}, a `partial domain axiom,' which, as Friedman and Moulin~\cite{FM} argue, is not very natural because it applies only to part of the space of initial and final values. 

In \cite{Owen}, Owen associates a multilinear function to any cooperative game so that applying the Aumann-Shapley method to this function yields the Shapley value of the game.  A generalization of these techniques may be used to prove Corollary~\ref{asirand} for path methods.  In this context, the full Corollary~\ref{asirand} may be viewed as a generalization to the case of an arbitrary attribution method.  Further, Theorem~\ref{converse} should be of independent interest to the cost sharing community because it identifies conditions on the characteristic function under which two important cost sharing methods, the Aumann-Shapley and Shapley-Shubik methods, can coincide, giving a converse to the result of~\cite{Owen} under slightly stronger conditions.\footnote{We note that our proof of Theorem~\ref{converse} does require the attribution context, however, as it relies crucially on the fact that attributions exist between any two pairs of values.}

\subsection{Notations}

We write $[n]$ for the set $\{1, 2, \ldots, n\}$.  For a set of variables $c_1, \ldots, c_n$ and a subset $I \subset [n]$, write $c$ for the $n$-tuple $(c_1, \ldots, c_n)$ and $c_I$ for the product $\prod_{i \in I} c_i$ over the indices in $I$. For two sets of variables $c_1, \ldots, c_n$ and $d_1, \ldots, d_n$, we write $c < d$ (resp. $c \leq d$) if $c_i < d_i$ (resp. $c_i \leq d_i$) for all $i$.  We write $[c, d]$ for the closed box $I_1 \times \cdots \times I_n$, where $I_i$ is the closed interval bounded by $c_i$ and $d_i$.  We use $0$ to denote the vector $(0, \ldots, 0)$ containing all $0$'s.  The length of this vector will always be clear from context.  For a function $f: \RR^n \to \RR$ and a multiset of indices $\alpha$, we denote by $\partial_\alpha f$ the mixed partial derivative with respect to the indices in $\alpha$.  In all cases where this construction appears, we will assume that $f$ is chosen so that Young's theorem on the equality of mixed partial derivatives holds.

\section{Applicability of our Model} \label{applications}

In this section we discuss practical applications of our model. For our model to be applicable, the characteristic function must be known, deterministic, and multilinear, and the values of the variables at the initial and final points must be known exactly. The examples in this section satisfy these properties, and are indicative of other settings in which our model is potentially applicable.

We begin with a few examples motivated by the Internet.

\begin{example} [Pay-per-click advertising \cite{ppc}]\label{ex:ss}
 The characteristic function is the spend $s$ of an advertiser, which can be expressed as the product $s = c \cdot p$ of the number of clicks $c$ that an advertiser's advertisement received and the average cost per click $p$. The final values of $s$, $c$, and $p$ are be statistics from a certain week, and the initial values are be statistics from the preceding week.  The problem then is to identify to what extent the  advertiser's change in spend is due to a change in the number of clicks versus a change in the cost per click.

A more granular spend model applicable in a specific form of pay-per-click advertising called sponsored search advertising is
\[
f_\text{spend} = q\cdot  b \cdot \sum_i p_i \cdot \mathrm{CTR}_i \cdot \mathrm{CPC}_i.
\]
Here, $q$ is the number of ad-views that the advertiser is eligible for, $b$ is the probability that the ads have sufficient budget to show, $p_i$ is the probability that an ad appears in the $i^\text{th}$ auction position, and $(\mathrm{CTR}_\text{i}, \mathrm{CPC}_\text{i})$ are the click through rate and the cost per click for the $i^\text{th}$ auction position.\footnote{Recall that all major search engines place some ads based on the results of an auction.}
\end{example}

\begin{example} [e-Commerce website analysis] \label{ex:website}
Consider an online retailer's website. We can model the website as a directed acyclic graph with a single sink $t$, which is the page displayed on a successful transaction (see Archak et al.~\cite{Archak} and Immorlica et al.~\cite{immorlica} for similar models). For every page, let $s_j$ denote the number of times that a surfer starts on page $j$. For every hyperlink directed from page $i$ to page $j$, let $p_{ij}$ denote the probability on average that a surfer follows this link given that he or she is at page $i$. The expected number of successful transactions is 
\[
\sum_{i \in V} s_i \sum_{\text{$P$ a path from $i$ to $t$}}\prod_{ (r,s) \in P } p_{rs},
\]
which is multilinear. The initial values for the variables are average statistics for the last year, and the final values are the same statistics for this year. The attributions to the variables $\{s_j\}$ and $\{p_{ij}\}$ may then yield insight into changes in traffic patterns that impact sales.
\end{example}

In our model, we require the characteristic function to be known and deterministic. This is in contrast to the fields of Regression Analysis~\cite{REG}, where the function and the inputs are statistical quantities and require model fitting and estimation, and structural equation modeling~\cite{SEM}, where additionally the variables may also require inference. Example~\ref{ex:ss} satisfies these conditions because the characteristic function models a software system whose working is known and deterministic; Example~\ref{ex:website} satisfies these conditions because the characteristic function models a flow through a known graph; and Example~\ref{ex:simple} from the introduction satisfies these conditions because it models a known supply chain.

Further, in both examples, we explain the change in performance of a system that occurred in the \emph{past}. Consequently, we have full-information about the initial and final values of the variables. Our model is not explicitly set up to be predictive about the future, though the insights gained could potentially be useful to guide future decisions. For instance, in Example~\ref{ex:ss}, advertisers who notice that a large negative impact is attributed to the budget variable $b$ may choose to raise their budgets. 

\begin{remark} 
One advantage of attribution is that it enables the comparison of changes in unlike quantities: For instance, in Example~\ref{ex:ss}, the advertiser can meaningfully compare the \emph{impact} of a change in the cost per click to the \emph{impact} of a change in the budget-related throttling rate ($b$). Such comparisons can aid decision making. In this example, the advertiser can decide if it is more worthwhile to focus on changing budgets, or controlling the cost per click (by changing bids in the ad auction).
\end{remark}

While the previous examples were about explaining the \emph{change} in the performance of a system, the next example, which is motivated by investment, involves comparing the performance of a system against a benchmark.
 
\begin{example} \label{ex:investment}
The performance of a portfolio can be expressed as the sum
\[
\sum_{i \in S} r_i \cdot w_i,
\]
where $r_i$ is the return within an asset class $i$ and  $w_i$ the amount invested within this asset class. Performance attribution~\cite{perf} attempts to explain why the performance of a portfolio (the final variables) deviates from the performance of a benchmark portfolio (the initial variables).  In particular, it asks whether the deviation in performance is due to the difference in the allocation of investments across asset classes (the attributions to the $w_i$'s) or to the selection of assets within an asset class (the attributions to the $r_i$'s). 

The standard way of doing performance attribution involves considering an active allocation term  $r_{i}^1 \cdot (w_{i}^2 - w_{i}^1)$, a security selection term $w_{i}^1 \cdot (r_{i}^2 - r_{i}^1)$, and a slack term $(r_{i}^2 - r_{i}^1) \cdot (w_{i}^2 - w_{i}^1)$ for each asset class; the latter term is necessary for completeness, but does not yield any insight.  In contrast, our approach yields completeness automatically.
\end{example}

Finally, here is an example from performance analysis of basketball statistics.

\begin{example}\label{ex:sports} 
Suppose the coaching staff of a basketball team wants insight into the change in offensive performance of the team from last year (the initial version of the variables) to this year (the final version of the variables).  Such studies are currently done in other frameworks as in~\cite{basketball2} or~\cite{basketball}.  Letting $n_i$, $m_i$, $a_i$, and $p_i$ be the number of games per season, the number of minutes per game, the number of attempts per minute, and the field goal percentage of each player, the total number of points scored by the team is
\[
f_\text{points} = \sum_{i} n_i \cdot m_i \cdot a_i \cdot \frac{p_i}{100}.
\]
Using attributions for $f_\text{points}$ in combination with other information can help the coaches understand and refine the performance of the team.
\end{example}

We now give two remarks illustrating some advantages of our attribution approach.

\begin{remark} 
A common way for humans to perform attribution relies on counterfactual intuition.  For instance, when we assert that smoking causes cancer, there is a presumption that holding all other things constant, not smoking will reduce the chance of contracting cancer. Such counterfactual semantics have been used as the basis for logics of causation (see Chapter 7 from Pearl~\cite{pearl}, for instance). 

Path methods (Definition~\ref{def:path-attribution}) and hence all the methods we consider in this paper have a natural counterfactual interpretation.  Every path method considers the counterfactual of moving between the initial and final values along the chosen family of paths.  Breaking this down, we may consider a path as the limit of piecewise linear paths which change only one independent variable at once.  From this viewpoint, the attribution to an independent variable is simply the cumulative change in the function due to this infinite number of infinitesimal counterfactuals.
\end{remark}

\begin{remark}
Let us reiterate the benefit of a method satisfying \ASI in light of the above examples. For many attribution problems, the units in which variables are measured are a matter of convention and are not canonical in any sense.  For instance, in pay-per-click advertising (Example~\ref{ex:ss}), the cost of advertising may be measured as the cost per thousand impressions or the cost per million impressions (see~\cite{ppc}), and, in basketball statistics, field goal accuracy is popularly expressed as a percentage between $0$ and $100$ rather than an accuracy rate between $0$ and $1$. 

In these examples, it critical that a different scaling of the units does not change the attribution. Specifying variables can be even more difficult than this, however.  Consider a characteristic function which depends on a dimensionless physical quantity such as the Reynolds number of a chaotic fluid or the Prandtl number of a material.  Such quantities lack natural units or even canonical reference points; for them, we would like the attribution to be invariant not only under rescaling of units but also changes of the zero points of these units. Such changes are exactly affine transformations, leading us to the \ASI axiom.
\end{remark}

We conclude this section with a discussion of the importance of applying attribution techniques carefully to yield meaningful insights.

\begin{remark} 
We return to the context of Example~\ref{ex:ss}. Besides sponsored search advertising, another common form of advertising is content advertising, that is, advertising on websites. Search ads commonly have a higher cost per click (CPC) than content ads because they are typically more contextual.  

Consider an advertiser who employs both forms of advertising, and who seeks to use attribution methods to analyze the impact of a change in the CPC on the amount of money it spends on advertising.  Suppose that the situation of the advertiser is summarized by Table~\ref{mix}.  Note that though the search CPC and the content CPC have both doubled, the overall CPC has actually fallen because of an increase in the proportion of clicks from content ads. 

There are two possible ways to perform attribution in this situation.  One way is to first compute the change in the overall CPC and use this to perform attribution (using the Aumann-Shapley-Shubik method, for instance).  Because the overall CPC fell, we would conclude that CPC's had a negative impact on spend.  Alternatively, if we reasoned about the impact of a change in the CPC of search ads and content ads separately, and then \emph{aggregated} the attributions, we would come up with the more meaningful conclusion that CPC's had a positive impact on the change in spend. Thus, \emph{aggregating the attributions} is more meaningful than \emph{attributing with aggregates} in this example.

\begin{table}[!ht]
\begin{center}
\begin{tabular}{ |l | c | c | c | c | c| c} \hline
   & Search CPC (\$) & Search Clicks & Content CPC (\$) & Content Clicks  & Overall CPC (\$) \\
  \hline
  Initial & 1 & 100 & 0.01 & 100 & 0.505 \\
  Final & 2 & 100 & 0.02 & 10000 & 0.0396\\ \hline
\end{tabular}
\end{center}

\medskip

\caption{An example with mix effects.}  \label{mix}
\end{table}
\end{remark}

\section{Characterizations of attribution methods for multilinear functions} \label{sec-main}

In this section, we seek an axiomatization for the class of multilinear functions. We focus on the class of multilinear functions for two reasons. First, this class of functions has several applications as illustrated in the previous section. Second, as discussed in Subsection~\ref{sec-comparison}, axiomatizations over a narrow family of functions can be more meaningful than axiomatizations that quantify widely over characteristic functions. 

We ignore additively separable functions for the rest of this section due to the following uniqueness result.

\begin{lemma} \label{addsepuni}
On additively separable functions, there is a unique attribution method that satisfies \ADD and \DUM.  
\end{lemma}
\begin{proof}
Write an additively separable function $f$ in the form $f(r_1, \ldots, r_n) = f_1(r_1) + \cdots + f_n(r_n)$.  By \ADD, for all $i$ we have
\[
z_i(r,s,f) = z_i(r,s,f_1) + \cdots + z_i(r,s,f_n).
\]
Now, by \DUM, $z_i(r,s,f_j) = 0$ for $j \neq i$, so by completeness $z_i(r,s,f_i) = f_i(s_i) - f_i(r_i)$, which implies that $z_i(r, s, f) = f_i(s_i) - f_i(r_i)$ is the unique attribution method on additively separable functions.
\end{proof}

\subsection{Methods that satisfy \ADD, \DUMP, and \CON} \label{adcmethods}

In this section we characterize attribution methods that satisfy the basic axioms \ADD, \DUMP, and \CON for multilinear characteristic functions.  

\begin{theorem} \label{pointwise}
Any attribution method on multilinear functions which satisfies \ADD, \DUMP, and \CON is a value-variant random order method.
\end{theorem}
\begin{proof}
Let $V$ be the space of multilinear functions on $x_1, \ldots, x_n$, and note that the monomials $x_I$ give a basis of $V$. Fix some $r, s$, and let $K$ be the set of indices $k$ such that $r_k \neq s_k$.    It will suffice to show that the attributions given by $z_i(r, s, -)$ correspond to the attributions of some value-variant random order method on $(r, s)$.  

\noindent \textit{Step 1: The space of attribution methods} 

Fix an attribution method $z_i$ satisfying \ADD, \DUMP, and \CON, and define $z_{i, I} := z_i(r, s, x_I)$.  By Lemma \ref{cauchy} applied to $z_i$, \ADD and \CON together imply that $z_i(r, s, -)$ is linear and therefore uniquely determined by its values $z_{i, I}$ on the monomials $x_I$.   By \DUMP and completeness, we see that such $z_{i, I}$ satisfy
\begin{align} \label{zeq1}
z_{i, I} &= 0 \text{ if $i \notin I$,}\\ \label{zeq2}
z_{i, I} &= 0 \text{ for $i \notin K$,}\\ \label{zeq3}
z_{i, I} &= r_{I - K} z_{i, I \cap K} \text{, and}\\ \label{zeq4}
\sum_{i \in I \cap K} z_{i, I} &= s_I - r_I.
\end{align}
Here, (\ref{zeq1}) and (\ref{zeq2}) follow immediately from \DUMP, (\ref{zeq3}) follows by noting that 
\[
z_{i, I} - r_{I - K} z_{i, I \cap K} = z_i(r, s, x_I - r_{I - K} x_{I \cap K}) = 0
\]
by \ADD and \DUMP, and (\ref{zeq4}) follows from completeness.  Values of $z_{i, I}$ satisfying constraints (\ref{zeq1}) through (\ref{zeq4}) are completely determined by the values of $z_{i, I}$ with $i \in I$ and $I \subset K$.  In fact, setting $k = |K|$, they form an affine subspace $\AA'$ of dimension
\[
\sum_{I \subset K} \max\{|I| - 1, 0\} = \sum_{i = 1}^{k} \binom{k}{i} (i - 1) = k \sum_{i = 1}^k \binom{k - 1}{i - 1} - (2^k - 1) = k 2^{k-1} - 2^k + 1
\]
inside the space $\Z$ of all tuples $\{z_{i, I}\}$ with $i \in I, I \subset K$.

Considering $\CON$ gives the additional constraint 
\begin{equation} \label{zeq5}
\sum_I a_I z_{i, I} \text{ has the same sign as $(s_i - r_i)$ if $\sum_I a_I x_I$ is non-decreasing in $r_i$}.
\end{equation}
Any $\{z_{i, I}\}$ satisfying (\ref{zeq1}) through (\ref{zeq5}) gives rise to a valid set of attributions on $(r, s)$. Hence, an attribution method satisfying \ADD, \DUMP, and \CON is characterized on the pair of values $(r, s)$ by the closed subspace $\AA$ of $\AA'$ defined by
\[
\AA:= \left\{\{z_{i, I}\} \text{ for } i \in I, I \subset K \mid \sum_{i \in I} z_{i, I} = s_I - r_I \text{ and $z_{i, I}$ satisfy (\ref{zeq5})} \right\}.
\]

\noindent \textit{Step 2: The space of value-variant random order methods} 

Let us now characterize the functionals $z_i(r, s, -)$ given by value-variant random order methods.  The space of attributions on $(r, s)$ which can result from value-variant random order methods is specified by giving for each $J \subset K$ and $j \in J$ a weight $c_{j, J}$ so that
\begin{align} \label{ceq1}
\sum_{i \in I} c_{i, I} &= \sum_{i \notin I} c_{i, I \cup \{i\}} \text{ for $I \neq \emptyset, K$}, \\ \label{ceq2}
\sum_i c_{i, \{i\}} &= 1, \\ \label{ceq3}
\sum_i c_{i, K} &= 1, \text{ and} \\ \label{ceq4}
c_{j, J} &\geq 0.
\end{align}
Note here that (\ref{ceq3}) is implied by (\ref{ceq1}) and (\ref{ceq2}). The space of such methods is therefore the closed subset $\R$ defined by $(\ref{ceq4})$ lying within the affine space $\R'$ defined by $(\ref{ceq1})$ through $(\ref{ceq3})$ inside the space $\C$ of all tuples $\{c_{j, J}\}$ with $j \in J, J \subset K$.  Notice that $\R'$ has dimension at least 
\[
\sum_{i = 1}^k \binom{k}{i} i  - (2^k - 1) = k 2^{k-1} - 2^k + 1.
\]

\noindent \textit{Step 3: Mapping from value-variant random order methods to attribution methods} 

We now understand the map $\phi: \R \to \AA$ between value-variant random order methods and attribution methods; it will be induced by a linear map $\phi: \C \to \Z$.  For $I \subset K$ and $i \in I$, setting $z_{i, I}^{ro} := z_i^{ro}(r, s, x_I)$ and $I' = I - \{i\}$, the map $\phi$ is given explicitly by
\begin{multline} \label{zroeq}
z_{i, I}^{ro} = \sum_{J \ni i} c_{i, J} (s_{I \cap J} r_{I - J} - s_{I \cap J - \{i\}} r_{I - J \cup \{i\}}) \\
	= \sum_{J \ni i} c_{i, J} s_{I \cap J - \{i\}} r_{I - J} (s_i - r_i) 
	= \sum_{J' \subset K - \{i\}} c_{i, J' \cup \{i\}} s_{I' \cap J'} r_{I' - J'} (s_i - r_i).
\end{multline}
Because value-variant random order methods are attribution methods satisfying \ADD, \DUMP, and \CON, the resulting attributions satisfy the constraints (\ref{zeq1}) through (\ref{zeq5}).  

\noindent \textit{Step 4: Checking that $\phi$ is injective}

We now claim that $\phi$ is injective.  By (\ref{zroeq}), the map $\phi$ is given by a $k 2^{k-1} \times k 2^{k-1}$ matrix $\Phi$ such that
\begin{itemize}
\item $\Phi$ is block diagonal with $2^{k-1} \times 2^{k-1}$ blocks, and

\item the $i^\text{th}$ block $\Phi^i$ of $\Phi$ is indexed by subsets of $K - \{i\}$ and has entries
\[
\Phi^i_{I', J'} = (s_i - r_i) s_{I' \cap J'} r_{I' - J'},
\]
where $I', J' \subset K - \{i\}$.
\end{itemize}
We must check that $\Phi^i$ is non-singular for each $i$.  For this, we claim that 
\[
\det \Phi^i = \prod_j (s_j - r_j)^{2^{k-1}}. 
\]
Consider the matrices
\[
A_K := (s_{I \cap J} r_{I - J})_{I, J \subset K}.
\]
We will show by induction on $k$ that 
\[
\det A_K = \prod_{k \in K} (s_k - r_k)^{2^{k-1}},
\]
which obviously implies the desired.   In the base case $k = 1$, we see that
\[
A_K = \left(\begin{matrix} 1 & 1 \\ r_1 & s_1 \end{matrix}\right)
\]
and the conclusion is obvious.  Now suppose the statement for some $k$ and take some $K$ with $|K| = k + 1$.  Then, pick some $j \in K$ and set $K' = K - \{j\}$.  Then, placing $A_K$ into block form, we have
\[
\det(A_K) = \det \left(\begin{matrix} A_{K'} & A_{K'} \\ r_j A_{K'} & s_j A_{K'}\end{matrix}\right) = \det \left(\begin{matrix} A_{K'} & 0 \\ r_j A_{K'} & (s_j - r_j) A_{K'}\end{matrix}\right) = (s_j - r_j)^{2^{k}} \det(A_{K'})^2 = \prod_{i \in K} (s_i - r_i)^{2^k},
\]
completing the induction.

\noindent \textit{Step 5: Putting everything together}

We have now shown that $\phi$ is injective as a map $\C \to \Z$.  Therefore, $\phi: \R' \to \AA'$ is an injective linear map between affine spaces with $\dim \R' \geq \dim \AA'$, hence an isomorphism of affine spaces.  It remains only to show that this isomorphism restricts to $\R \to \AA$; for this, we match the conditions (\ref{zeq5}) and (\ref{ceq4}) to check that $\phi$ maps $\R' - \R$ to $\AA' - \AA$.  

For any $\{c_{j, J}\} \in \R' - \R$, choose $j$ and $J$ with $j \in J \subset K$ such that $c_{j, J} < 0$.  For $I \subset [n]$, define the points $u^I$ and $v^I$ by 
\[
u_i^I = \begin{cases} r_i & i \notin I \\ s_i & i \in I \end{cases} \text{ and } v_i^I = \begin{cases} s_i & i \notin I \\ r_i & i \in I. \end{cases}
\]
By Lemma \ref{mldelta}, we may find a multilinear function $h(x_1, \ldots, \widehat{x}_j, \ldots, x_n)$ so that $h(u_i^I) = \delta_{I, J}$.  Now, $h$ is linear in each $x_i$, hence it is non-negative on $[r, s]$ because it is non-negative on each of the vertices of $[r, s]$.  Therefore, the multilinear function  
\[
g(x) = x_j h(x_1, \ldots, \widehat{x}_j, \ldots, x_n)
\]
satisfies $g(u^J) - g(u^{J - \{j\}}) = s_j - r_j$ and $g(u^I) = g(u^{I - \{j\}})$ for all $I \neq J$.  Further, because $h$ is non-negative on $[r, s]$, $g$ is non-decreasing in $x_j$ on $[r, s]$.  On the other hand, we see that 
\[
z^{ro}_j(r, s, g) = \sum_{I \ni j} c_{j, I} (g(u^I) - g(u^{I - \{j\}})) = c_{j, J} (g(u^J) - g(u^{J - \{j\}})) = c_{j, J}(s_j - r_j),
\]
which has opposite sign from $s_j - r_j$.  This means that the image of $\{c_{j, J}\} \in \R' - \R$ under $\phi$ does not satisfy $\CON$, hence lies in $\AA' - \AA$.  Therefore, we conclude that $\phi$ maps $\R' - \R$ to $\AA' - \AA$, hence $\phi$ maps $\R$ bijectively to $\AA$, as needed.
\end{proof}

\begin{remark}
In the proof of Theorem \ref{pointwise}, \CON plays two different roles.  First, it provides the technical condition that allows us to convert from \ADD to linearity by using Lemma \ref{cauchy}.  Secondly and more crucially, it is necessary because any value-variant random order method is a convex combination of the attributions along the paths $\gamma_{r, s}^\sigma$ rather than an affine combination.  As a result, such a method satisfies \CON.
\end{remark}

Recall from Section~\ref{sec-path} that path attribution methods satisfy \ADD, \DUMP, and \CON for all characteristic functions.  Therefore, for multilinear functions, Theorem \ref{pointwise} implies that all path attribution methods are value-variant random order methods. This is somewhat surprising because random order methods are inherently combinatorial and may be evaluated using the values of the characteristic function at a finite set of points, while path attribution methods require in general a continuous evaluation of the characteristic function. Thus we see that the form of the characteristic function in Theorem~\ref{pointwise} is key in reducing the latter continuous evaluation to a discrete one. See Section~\ref{sec-ass} for an explicit illustration of this in the context of the Aumann-Shapley method.

\subsection{Methods that satisfy \ASI}

The characterization in the previous section allows for significant freedom in the selection of an attribution method, arguably undesirably so. For instance, it is possible to vary the convex combination over the random order paths for each $r,s$ in some discontinuous way. To address this issue, we impose in this subsection a continuity condition on our paths. Following our axiomatic approach, we would like to impose this continuity condition on paths via an axiom on our attribution methods.  A natural candidate, then, is \ASI, as it is a continuity condition on attributions and has a very natural interpretation in the attribution context. With the addition of \ASI, we have the following.

\begin{cor} \label{asirand}
Any attribution method on multilinear functions satisfying \ADD, \DUM, \CON, and \ASI is a random order method.
\end{cor}
\begin{proof}
First, we claim that \DUM and \ASI imply \DUMP for $r, s$ such that $r_i \neq s_i$ for any $i$.  Suppose that a characteristic function $f$ does not depend on the value of $r_i$ on $[r, s]$.  We may write
\[
f(r_1, \ldots, r_n) = f^1(r_1,\ldots, \hat{r}_i, \ldots, r_n) + f^2(r_1,\ldots, \hat{r}_i, \ldots, r_n) r_i,
\]
where $f^2(r_1,\ldots, \hat{r}_i, \ldots, r_n) = 0$ on $[r, s]$, which is Zariski dense in $\RR^{n}$, hence $f^2 = 0$ as a polynomial.  This implies that $f = f^1$, so the result holds by \DUM.  

Now, take $r^* = (0, \ldots, 0)$ and $s^* = (1, \ldots, 1)$.  By Theorem \ref{pointwise}, we see that
\[
z_i(r^*, s^*, f) = z_i^*(r^*, s^*, f)
\]
for some random order method $z_i^*$.  Because $z_i$ and $z_i^*$ both satisfy $\ASI$, this implies that $z_i = z_i^*$ is a random order method, as needed.
\end{proof}

\subsection{Main Result}

The characterization in the previous section allows us to treat independent variables asymmetrically. For instance, we could consider only a single random order path in our convex combination. But there appears no \textit{a priori} reason to treat variables asymmetrically, and so we impose \ANON, which gives us the following axiomatization.

\begin{cor}\label{unique-ASI-Anon}
There is a unique attribution method on multilinear functions satisfying \ADD, \DUM, \CON, \ASI, and \ANON.
\end{cor}
\begin{proof}
By Corollary \ref{asirand}, such a method must be a random order method.  But there is a unique random order method satisfying \ANON, the Shapley-Shubik method, as needed.
\end{proof}

\section{The Aumann-Shapley-Shubik method} \label{sec-ass}

Recall from Section~\ref{shapley-based} that the Aumann-Shapley method satisfies all the axioms mentioned in Theorem~\ref{unique-ASI-Anon} for \emph{every} characteristic function, while Corollary~\ref{unique-ASI-Anon} shows that there is a unique method that satisfies these axioms for multilinear functions. This implies that the Aumann-Shapley method coincides with the Shapley-Shubik method for multilinear functions. We note that a proof by direct computation is also possible; for completeness, we show this proof in Appendix \ref{pappend}.

\begin{theorem} \label{AS-SS-agree}
If $f$ is the sum of a multilinear function and an additively separable function, then the Aumann-Shapley (Definition~\ref{AS}) and Shapley-Shubik (Definition~\ref{SS}) attribution methods agree for $f$.
\end{theorem}

We illustrate the attributions that Aumann-Shapley and Shapley-Shubik yield on small instances of multilinear functions in the following example. 

\begin{example}
For $f(r_1, r_2) = r_1 r_2$, these methods coincide and both methods give:
\[
z_1(r, s, f) = (s_1 - r_1) \frac{r_2 + s_2}{2} \text{ and } z_2(r, s, f) = (s_2 - r_2) \frac{r_1 + s_1}{2}.
\]
In particular, when $r = 0$, both methods correspond to an equal split. For $f(r_1, r_2, r_3) = r_1 r_2 r_3$, the attributions again agree and are
\begin{align*}
z_1(r, s, f) &= (s_1 - r_1)\frac{2 r_2 r_3 + 2 s_2 s_3 + r_2 s_3 + s_2 r_3}{6}, \\
 z_2(r, s, f) &= (s_2 - r_2)\frac{2 r_1 r_3 + 2 s_1 s_3 + r_1 s_3 + s_1 r_3}{6}, \text{ and}\\
 z_3(r, s, f) &= (s_3 - r_3)\frac{2 r_1 r_2 + 2 s_1 s_2 + r_1 s_2 + s_1 r_2}{6}.
\end{align*}
\end{example}

We may now define the \emph{Aumann-Shapley-Shubik} method for characteristic functions that are the sum of a multilinear and an additively separable function as the method equivalent to both the Aumann-Shapley and Shapley-Shubik methods. Summarizing the conclusions of Theorem~\ref{AS-SS-agree} and Corollary~\ref{unique-ASI-Anon}, we obtain the following axiomatic characterization of the Aumann-Shapley-Shubik method.

\begin{cor} \label{axiom-ASS}
For characteristic functions $f$ which are the sum of a multilinear function and an additively separable function, the Aumann-Shapley-Shubik method is the unique method satisfying \ADD, \DUMP, \CON, \ANON and \ASI.
\end{cor}

\begin{remark}
Corollary~\ref{unique-ASI-Anon} and the fact that the Shapley-Shubik method satisfies \MON together imply that the Aumann-Shapley-Shubik method satisfies \MON. Further, Sprumont and Wang~\cite{SW} show that the Shapley-Shubik method satisfies a property stronger than \ASI called \axiom{Ordinal Invariance}, meaning that the Shapley-Shubik method is invariant under all order-preserving (monotone) reparameterizations of the variables. Corollary~\ref{unique-ASI-Anon} implies that this carries over to the Aumann-Shapley-Shubik method.
\end{remark}

\subsection{When do Aumann-Shapley and Shapley-Shubik agree?}

Having identified the Aumann-Shapley-Shubik method as a uniquely desirable one for characteristic functions which are the sum of a multilinear function and an additively separable function, we now consider when it exists.  As we show in the following Theorem~\ref{converse}, this will occur only if the characteristic function $f$ takes this form.

\begin{theorem} \label{converse}
If the Aumann-Shapley and Shapley-Shubik attribution methods agree for some cost function $f$, then $f$ is the sum of a multilinear function and an additively separable function.
\end{theorem}
\begin{proof}
By Lemma~\ref{char}, it suffices for us to show that $\partial_{iij} f = 0$ for distinct $i, j$.  We first consider the case $n = 2$, in which case we wish to show that $\partial_{12}f$ is constant.  Then, for any $r = (r_1, r_2)$ and $s = (s_1, s_2)$ with $r \leq s$, the Aumann-Shapley attribution to the second variable is
\[
z_2^{AS}(r, s, f) = \int_0^1 \partial_2f(\gamma(t)) \gamma_2'(t) dt
\]
with $\gamma(t) = (1 - t)r + t s$.  On the other hand, the Shapley-Shubik attribution is 
\[
z_2^{SS}(r, s, f) = \frac{1}{2}[f(s_1, s_2) - f(s_1, r_2)] + \frac{1}{2}[f(r_1, s_2) - f(r_1, r_2)].
\]
Subdivide the rectangle $R$ with vertices at $(r_1, r_2)$, $(r_1, s_2)$, $(s_1, r_2)$, and $(s_1, s_2)$ into the triangular regions $T_1$ lying above the path of $\gamma$ and $T_2$ lying below the path of $\gamma$ as shown in Figure~\ref{figa} below.

\begin{figure}[!ht]
\begin{center}
\subfigure[Regions $T_1$ and $T_2$]{\label{figa}\includegraphics[height=1.34in]{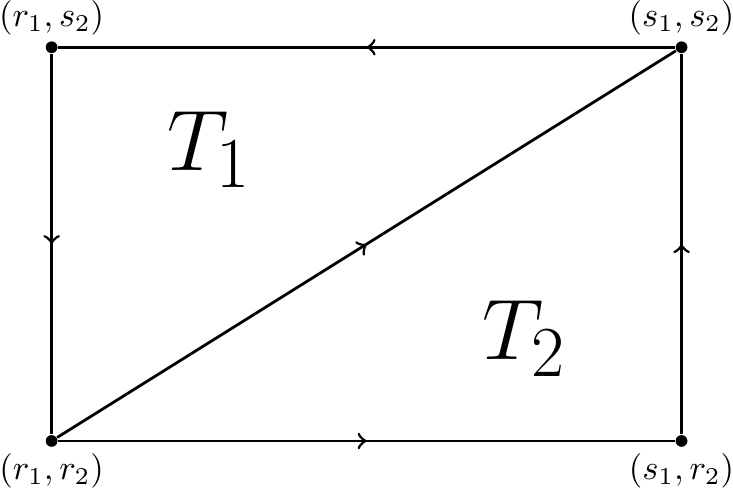}}
\subfigure[Removing two pairs of triangles.]{\label{figb}\includegraphics[height=1.34in]{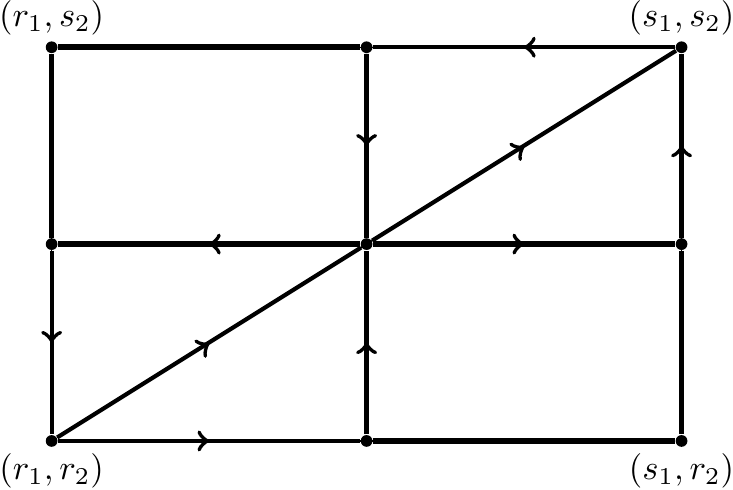}}
\subfigure[Rectangles to apply (\ref{eq2}) on.]{\label{figc}\includegraphics[height=1.34in]{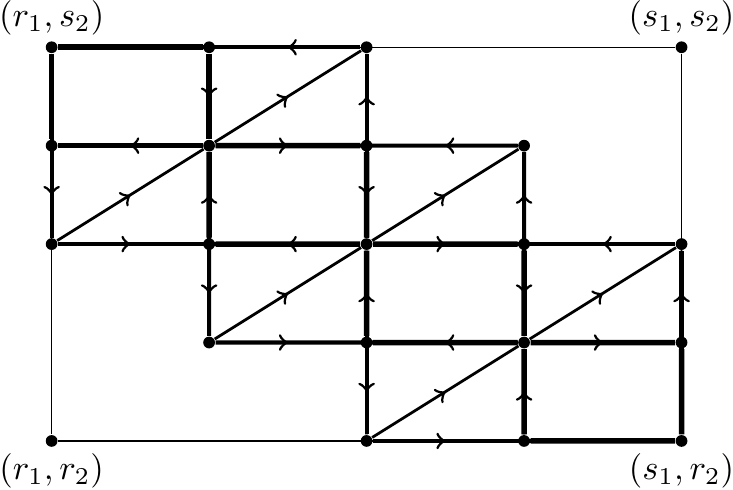}}
\end{center}
 \caption{Steps in the proof of Theorem~\ref{converse}}
  \label{pfstep}
\end{figure}

Then, by Stokes' Theorem, we have
\[
\int_{T_1} \partial_{12} f(x_1, x_2) dx_1 dx_2 = \int_{\partial T_1} \partial_2 f(x_1, x_2) dx_2 =  \int_0^1 \partial_2 f(\gamma(t)) \gamma_2'(t) dt - [f(r_1, s_2) - f(r_1, r_2)]
\]
and
\[
\int_{T_2} \partial_{12} f(x_1, x_2) dx_1 dx_2 = \int_{\partial T_2} \partial_2 f(x_1, x_2) dx_2 = [f(s_1, s_2) - f(s_1, r_2)] - \int_0^1 \partial_2 f(\gamma(t)) \gamma_2'(t) dt.
\]
Because $z_2^{SS}(r, s, f) = z_2^{AS}(r, s, f)$ by assumption, subtracting the two previous equations and applying our previous computations gives that
\begin{equation} \label{eq1}
\int_{T_1} \partial_{12} f(z_1, z_2) dz_1 dz_2 = \int_{T_2} \partial_{12} f(z_1, z_2) dz_1 dz_2
\end{equation}
for any choice of $r, s$. In particular, applying (\ref{eq1}) for the pairs $(r, s)$, $\left(r, \frac{r + s}{2}\right)$, and $\left(\frac{r + s}{2}, s\right)$ and subtracting the result of the latter two from the first, we obtain 
\begin{equation} \label{eq2}
\int_{\left[r_1, \frac{r_1 + s_1}{2}\right] \times \left[\frac{r_2 + s_2}{2}, s_2\right]} \partial_{12} f = \int_{\left[\frac{r_1 + s_1}{2}, s_1\right] \times \left[r_2, \frac{r_2 + s_2}{2}\right]} \partial_{12} f
\end{equation}
for all $r, s$.  The results of this process are shown in Figure~\ref{figb}.  Now, for any $x = (x_1, x_2)$, set $x' = (x_1, -x_2)$.  Applying (\ref{eq2}) to the pairs $(r, r + 2x), (r + x', r + x' + 2x), \ldots, (r + n x', r + n x' + 2x)$, we find that for any $n$ we have
\begin{equation} \label{eq3}
\int_{[r_1, r_1 + x_1] \times [r_2 + x_2, r_2 + 2x_2]} \partial_{12} f = \int_{[r_1 + (n+1)x_1, r_1 + (n+2)x_1] \times [r_2 - nx_2, r_2 - (n-1)x_2]} \partial_{12} f.
\end{equation}
This process is shown in Figure~\ref{figc}.

Suppose now for the sake of contradiction that $\partial_{12} f$ were not constant.  Then, there must exist some $r < s$ such that $\partial_{12}f(r) \neq \partial_{12} f(s)$.  Suppose without loss of generality that $\partial_{12} f(r) > \partial_{12} f(s)$.  Because $\partial_{12} f$ is continuous, we may find open neighborhoods $U$ of $r$ and $V$ of $s$ such that $\partial_{12} f(x) > \partial_{12} f(y)$ for $x \in U, y \in V$.   Now, choose $x = (x_1, x_2)$ and $n$ so that $[r_1, r_1 + x_1] \times [r_2 + x_2, r_2 + 2x_2] \subset U$ and that $[r_1 + (n+1)x_1, r_1 + (n+2)x_1] \times [r_2 - nx_2, r_2 - (n-1)x_2] \subset V$, in which case (\ref{eq3}) provides a contradiction.  Therefore, $\partial_{12} f$ is constant, which completes the proof in the case $n = 2$.

For the general case, choose any two variables $r_i$ and $r_j$.  Restricting to attributions between points with all other variables held fixed, the $n = 2$ case tells us that $\partial_{ij} f$ is independent of $r_i$ and $r_j$, which means exactly that $\partial_{iij} f = 0$ and $\partial_{ijj}f = 0$.  This holds for all $i, j$, so $f$ takes the desired form.
\end{proof}

One implication of Theorem \ref{converse} is that we will need a different axiomatization for characteristic functions that are not the sum of an additive and a multilinear function. Additionally, it justifies our restriction to sums of multilinear and additive characteristic functions.

\begin{remark}
The proof of Theorem \ref{converse} relied heavily on Stokes' theorem.   In fact, this approach works more generally to compare general path attribution methods; we summarize the idea briefly here.  Consider a single-path attribution method corresponding to a family of paths $\gamma_{r,s}$.  Letting $I_{r, s}$ be the (closed) image of $\gamma_{r, s}$ in $[r, s]$, we see that the attributions are given by 
\begin{equation} \label{attrib-line}
z_i(r, s) = \int_0^1 \partial_i f(\gamma_{r, s}(t)) \gamma_{r, s, i}'(t) dt = \int_{I_{r, s}} \partial_i f(r) dr_i,
\end{equation}
where we view $\partial_i f(r)\, dr_i$ as a differential form on $I_{r, s}$.  From this perspective, it is clear that $z_i(r, s)$ depends only on the underlying set $I_{r, s}$ of the path and not on the choice of parametrization $\gamma_{r, s}$.  

We can now use this viewpoint to compare methods.  Consider the case $n = 2$.  Let $\gamma^1_{r, s}$ and $\gamma^2_{r, s}$ be two families of paths and consider the corresponding single-path attribution methods.  If these methods coincide for some characteristic function $f$, then for all $r, s$, we have for all $i$ that 
\[
z_i(r, s, f) = \int_{I^1_{r, s}} \partial_i f dr_i = \int_{I^2_{r, s}} \partial_i f dr_i,
\]
where $I^1_{r, s}$ and $I^2_{r, s}$ are the images in $[r, s]$ of $\gamma_{r, s}^1$ and $\gamma_{r, s}^2$, respectively.  Suppose for simplicity that the closed curve formed by first traversing $\gamma_{r, s}^1$ and then traversing $\gamma_{r, s}^2$ is not self-intersecting.  Then, it bounds an open set $A_{r, s}$ in $[r, s]$.  From (\ref{attrib-line}), we then find that 
\begin{equation} \label{attrib-2d}
0 = \int_{I^1_{r, s}} \partial_2 f dr_2 - \int_{I^2_{r, s}} \partial_2 f dr_2 = \int_{A_{r, s}} \partial_{12}f dr_1 dr_2,
\end{equation}
where the final equality follows from Stokes' Theorem.  We have therefore translated condition (\ref{attrib-line}) involving line integrals to condition (\ref{attrib-2d}) involving area integrals.  In the situation of Theorem \ref{converse}, this condition is (\ref{eq1}), which we may analyze by elementary means because $T_1$ and $T_2$ are geometrically quite simple.  The general case seems to require different techniques; some ongoing work in this direction by the authors uses an approach involving tools from wavelet theory.
\end{remark}

\subsection{Computing Aumann-Shapley-Shubik}

In this subsection, we discuss the efficient computation of the Aumann-Shapley-Shubik method for multilinear functions.\footnote{We ignore additively separable functions because the attribution assigned to a variable is simply the change in the function in which it appears.} As discussed at the end of Subsection~\ref{adcmethods}, if $f$ is a multilinear function, then this method is computable in finite time because it coincides with the Shapley-Shubik method. Indeed, the attributions given by the Shapley-Shubik method are the average of the marginal impact of changing a variable over the finite number of possible variable orderings. However, there does not always exist an efficient (polynomial time) algorithm to compute the Shapley-Shubik attributions (see the hardness results in~\cite{Deng, Matsui}).

Now, for $f(r) = r_1 \cdot \cdots \cdot r_n$, the most basic example of a multilinear function, the Aumann-Shapley-Shubik attributions $z_i(r, s, f)$ are computable in finite time, as to compute the Shapley-Shubik attributions in this case it suffices to evaluate $f$ a finite number of times.  In principle, this may involve $\Theta(2^n)$ evaluations, one for each of the vertices of $[r, s]$.  However, Theorem~\ref{compute} below implies that in this case we may compute attributions in time quadratic in the number of variables.  If we instead consider general multilinear functions, iterating the algorithm of Theorem~\ref{compute} in Corollary~\ref{compute-iterate} yields runtime quadratic in the number of variables and linear in the number of non-zero monomials in the characteristic function.  These two results together ensure that our attribution theory is not impractical for computational reasons.

\begin{theorem} \label{compute}
Let $f(r) = r_1\cdots r_n$.  Then, for any $r, s$ and each $i$, the Aumann-Shapley-Shubik attribution $z_i(r, s, f)$ is computable in $O(n^2)$ time and $O(n)$ memory.
\end{theorem}
\begin{proof}
From the calculations in the computational proof of Theorem~\ref{AS-SS-agree} given in Appendix \ref{pappend}, the attributions are given by
\begin{align*}
z_i(r, s, f) &= \frac{1}{n!} (s_i - r_i) \sum_{K \subset [n] - \{i\}} |K|!(n - 1 - |K|)! s_K r_{[n] - \{i\} - K} \\
	&= \frac{1}{n!} (s_i - r_i) \sum_{k = 0}^{n-1} k! (n - 1 - k)! \sum_{\substack{K \subset [n] - \{i\} \\ |K| = k}} s_K r_{[n] - \{i\} - K},
\end{align*}
so it suffices to compute this value.  The computation is invariant under relabeling of coordinates, so we may assume for convenience of notation that $i = n$.  In this case, we have
\[
z_n(r, s, f) = \frac{1}{n!} (s_n - r_n) \sum_{k = 0}^{n-1} k! (n - 1 - k)! \sum_{\substack{K \subset [n-1] \\ |K| = k}} s_K r_{[n - 1] - K}.
\]
Our approach is to compute the sums
\[
X_{k, m} := \sum_{\substack{K \subset [m] \\ |K| = k}} s_K r_{[m] - K}
\]
for $m \leq n-1$ and $0 \leq k \leq m$ using dynamic programming.  Computing $z_i(r, s, f)$ then requires only a simple summation.  Algorithm~\ref{AS-SS-compute} formalizes this idea.

\begin{algorithm}[!ht]
\caption{Computing the Aumann-Shapley-Shubik attribution $z_n(r, s, f)$.}
\label{AS-SS-compute}
\begin{algorithmic}
\STATE $X_{0, 0} \gets 1$
\FOR{$m = 1$ to $n-1$}
	\STATE $X_{0, m} \gets r_m \cdot X_{0, m-1}$
	\FOR{$k = 1$ to $m-1$}
		\STATE $X_{k, m} \gets s_m \cdot X_{k-1, m-1} + r_m \cdot X_{k, m-1}$
	\ENDFOR
	\STATE $X_{m, m} \gets s_m \cdot X_{m-1, m-1}$
\ENDFOR
\RETURN $\frac{1}{n!}(s_n - r_n) \sum_{k = 0}^{n-1} k! (n - 1 - k)! \cdot X_{k, n-1}$
\end{algorithmic}
\end{algorithm}

The correctness of Algorithm~\ref{AS-SS-compute} follows from the evident recursion
\[
X_{k, m} = \begin{cases} r_m \cdot X_{0, m-1} & k = 0 \\ s_m \cdot X_{k-1, m-1} + r_m \cdot X_{k, m-1} & 1 \leq k \leq m -1 \\ s_m \cdot X_{m-1, m-1} & k = m \end{cases}
\]
and the expression for $z_i(r, s, f)$ obtained at the beginning of the proof.  There are $O(n^2)$ iterations of the loop, each taking $O(1)$ time to update $X_{k,m}$, giving a total runtime of $O(n^2)$.  Further, at each step, only the values of $X_{k, m}$ for $0 \leq k \leq m$ and $X_{k, m-1}$ for $0 \leq k \leq m-1$ are required; storing only these yields a memory requirement of $O(n)$.
\end{proof}

\begin{cor} \label{compute-iterate}
Let $f$ be a multilinear characteristic function in $n$ variables with $N$ non-zero monomial terms.  Then, the Aumann-Shapley-Shubik attribution $z_i(r, s, f)$ is computable in $O(n^2\cdot N)$ time and $O(n)$ memory.
\end{cor}
\begin{proof}
By \ADD and \DUM, we may simply run the algorithm of Theorem~\ref{compute} $N$ times, once for each non-zero monomial in $f$, and sum the resulting contributions.  This trivially gives the desired runtime and memory costs.
\end{proof}

\bibliographystyle{acm}
\bibliography{attribution}

\appendix

\section{A review of Stokes' theorem} \label{stokes}

In this appendix, we give a brief intuitive introduction to Stokes' theorem as it relates to our paper for the unfamiliar reader.  To minimize technical difficulties, we restrict ourselves to the case of dimension two, where Stoke's Theorem coincides with Green's Theorem, and suppress technical assumptions.  First, we state a basic version of the theorem.  

\begin{theorem}[Stokes' Theorem] \label{st}
Let $A$ be the region enclosed by a smooth closed curve in the plane.  Let $f$ be a differentiable function defined on an open neighborhood of $A$, and let $\partial A$ be the (oriented) boundary of $A$.  Then, we have
\begin{equation} \label{ststate}
\int_{\partial A} f dx_2 = \int_A \partial_1 f dx_1 dx_2.
\end{equation}
\end{theorem}

Let us explain intuitively the meaning of Theorem~\ref{st}.  It relates the path integral of the $1$-dimensional differential form $f dx_1$ along the boundary $\partial A$ of $A$ to the double integral of its exterior derivative $d(fdx_1) = \partial_2 f dx_1 dx_2$ on the interior of $A$.  We may visualize this in Figure~\ref{figstokes} below.

\begin{figure}[!ht]
\includegraphics[height=1.5in]{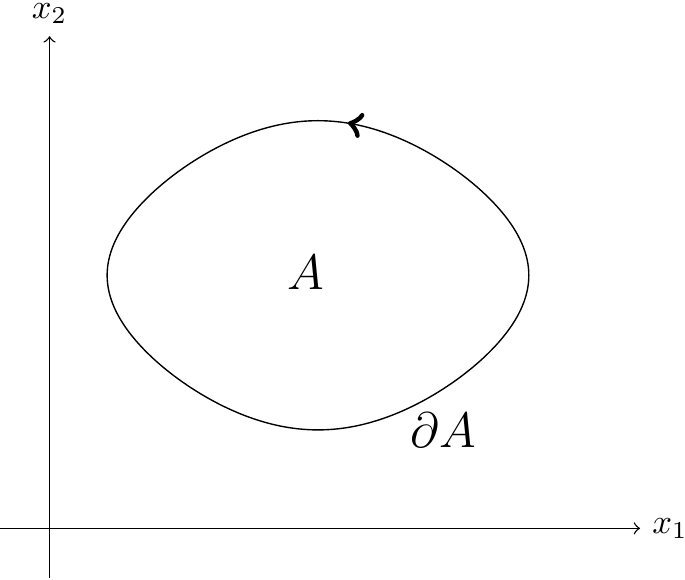}
\caption{A region $A$ and its boundary $\partial A$ in Stokes' theorem.} 
\label{figstokes}
\end{figure}

It may be instructive to consider an analogy between Stokes' theorem and the fundamental theorem of calculus (which is actually Stokes' theorem in dimension $1$). For a differentiable function $F$, the fundamental theorem of calculus relates the integral of $F'(x)$ along an interval to the difference in values of $F$ on the endpoints of this interval.  That is, it states that
\[
F(b) - F(a) = \int_a^b F'(x) dx.
\]
Stokes' theorem generalizes the fundamental theorem of calculus in the sense that it replaces the concept of an interval with a simple region, and the endpoints of the interval (which form its boundary) with the closed curve that forms the boundary of the region. Its proof is also ultimately an application of the fundamental theorem of calculus. We refer the interested reader to Chapter 11 of \cite{Apostol} or to \cite{Warner} for more detailed expositions of Stokes' theorem, which also appears in various engineering applications such as electrostatics and fluid dynamics. 

In this paper, Stokes' theorem is particularly convenient because it allows manipulation of line integrals of $1$-dimensional differential forms.  We see that the attributions given by path attribution methods take exactly this form for differential forms involving the characteristic function. Applying Stokes' theorem now yields conditions on the area integral of a mixed partial which we use as a starting point for further considerations.



\section{Technical results on multilinear functions}

In this appendix we state and prove some technical results about multilinear functions which are used in our proofs.  We begin with an alternate characterization of functions which are the sum of a multilinear function and an additively separable function.

\begin{lemma} \label{char}
A function $f: \RR^n \to \RR$ is the sum of a multilinear function and an additively separable function if and only if $\partial_{iij} f = 0$ for all $i \neq j$.
\end{lemma}
\begin{proof}
It is obvious that the sum of a multilinear function and an additively separable function has this property, so it remains to show the converse. We proceed by induction on $n$, with the base case $n = 1$ trivial.  Now, if $n > 1$, we may write $\partial_{11} f = g_1$ as a function of $q_1$ only, hence we see that 
\[
\partial_1 f(r) = \int_0^{r_1} g_1(t) dt + h(r)
\]
and
\[
f(r) = \int_0^{r_1} \int_0^{t_2} g_1(t_1) dt_1 dt_2 + r_1 h(r) + p(r),
\]
where $\partial_1h = \partial_1p = 0$. It remains to show that $h$ is multilinear and that $p$ is the sum of a multilinear function and an additively separable function.  Now, for any distinct $i, j \neq 1$, we have that 
\[
0 = \partial_{iij} f = r_1 \partial_{iij} h + \partial_{iij} p,
\]
so taking $r_1 = 0$ shows that $\partial_{iij} p = 0$.  Hence $p$ is the sum of a multilinear function and an additively separable function by the inductive hypothesis.  Now, notice that for $i \neq 1$, we have
\[
0 = \partial_{ii1} f = \partial_{ii} h,
\]
so $h$ is multilinear.  This completes the induction.
\end{proof}

\begin{remark}
The condition in Lemma~\ref{char} is a mixture of the conditions for $f$ to be multilinear ($\partial_{ii} f = 0$) and to be additively separable ($\partial_{ij} = 0$).
\end{remark}

Now, fix $r, s \in \RR^n$.  Our next two results provide alternate bases for the space of multilinear functions which will be convenient in analyzing their restriction to the vertices of $[r, s]$.  For $I \subset [n]$, define the points $u^I$ and $v^I$ by 
\[
u_i^I = \begin{cases} r_i & i \notin I \\ s_i & i \in I \end{cases} \text{ and } v_i^I = \begin{cases} s_i & i \notin I \\ r_i & i \in I. \end{cases}
\]
Notice that $u^I = v^{[n] - I}$ and that the vertices of $[r, s]$ are exactly the points $u^I$ as $I$ ranges over the subsets of $[n]$.  We then have the following two characterizations of multilinear functions.

\begin{lemma} \label{mldelta}
For any $r,s \in \RR^n$ with $r_i \neq s_i$ for all $i$, there exists for any $I \subset [n]$ a multilinear function $g_I$ such that $g_I(u^J) = \delta_{I, J}$.
\end{lemma}
\begin{proof}
Define $g_I$ by 
\[
g_I(x) = \frac{\prod_{i \in I} (x_i - v_i^I)}{\prod_{i \in I} (u_i^I - v_i^I)}.
\]
For any $I \neq J$, there is some $i \in K$ such that $u_i^I = v_i^J$, meaning that $g_I(u^J) = 0$ for $I \neq J$.  But $g_I(u^I) = 1$ by definition, so this $g_I$ has the desired properties.
\end{proof}

\begin{lemma} \label{mlbasis}
For any $r, s \in \RR^n$ and any $x_j$, there is a basis $\{f_\alpha\}$ of the space of multilinear functions such that $f_{\alpha}$ is non-decreasing in $x_j$ on $[r, s]$.
\end{lemma}
\begin{proof}
It suffices to consider the case where $r_i \neq s_i$ for all $i$, as otherwise we may pick $r', s'$ with $[r, s] \subset [r', s']$ and $r_i' \neq s_i'$ for all $i$.  Further, we may assume that $r_i \leq s_i$ for all $i$, as otherwise we may simply exchange $r_i$ and $s_i$. Now, for $J \subset [n] - \{j\}$, consider the multilinear functions
\[
g_{J \cup \{j\}} \text{ and } g_{J \cup \{j\}} + g_J
\]
given by Lemma \ref{mldelta}.  It is clear that these form a basis for the space of all multilinear functions because $\{g_I\}_{I \subset [n]}$ does.  Further, notice that 
\[
\partial_j g_{J \cup \{j\}} \text{ and } \partial_j (g_{J \cup \{j\}} + g_J)
\]
are multilinear functions in $x_1, \ldots, \widehat{x}_j, \ldots, x_n$ which are non-negative on all vertices of $[r, s]$, hence non-negative on $[r, s]$.  Therefore, they give the desired basis for the space of multilinear functions consisting of functions non-decreasing in $x_j$ on $[r, s]$.  
\end{proof}

The existence of the basis of Lemma \ref{mlbasis} allows us to convert \ADD to linearity as follows.

\begin{lemma} \label{cauchy}
Fix $r, s\in \RR^n$ and a variable $x_j$, and let $\phi$ be an additive functional on the space of multilinear functions.  If $\phi(f) \geq 0$ for $f$ non-decreasing in $x_j$ on $[r, s]$, then $\phi$ is linear.
\end{lemma}
\begin{proof}
Let $\{f_\alpha\}$ be the basis given by Lemma \ref{mlbasis}.  By additivity, it suffices to check that $\phi$ is linear on $\spann(f_\alpha)$ for each $\alpha$.  But $\phi(c f_\alpha) \geq 0$ for any $c \geq 0$, hence $\phi$ is additive and non-decreasing on $\spann(f_\alpha)$.  It is therefore linear on $\spann(f_\alpha)$ as a monotone solution to the Cauchy functional equation (see \cite{AE} or the original paper of \cite{Dar}).
\end{proof}

\section{Proof of Theorem \ref{AS-SS-agree}} \label{pappend}

In this appendix, we give a computational proof of Theorem \ref{AS-SS-agree}, which was omitted from the main text to streamline the exposition.  First we need a technical lemma.

\begin{lemma} \label{binom-int}
For non-negative integers $i, j$, we have 
\[
\int_0^1 x^i (1-x)^j dx = \frac{1}{(i + j + 1) \binom{i+j}{i}} dx.
\]
\end{lemma}
\begin{proof}
We induct on $i$.  For $i = 0$, the result is clear.  Now, suppose that the result holds for some $i-1$.  In this case, integration by parts gives that
\begin{align*}
\int_0^1 x^i (1 - x)^j dx &= \left[\frac{1}{j+1} x^i (1 - x)^{j+1}\right]_0^1 + \int_0^1 i x^{i-1} \frac{1}{j+1} (1 - x)^{j+1} dx \\
&= \frac{i}{j+1} \int_0^1 x^{i-1} (1-x)^{j+1} dx \\
&= \frac{i}{j+1} \frac{1}{(i + j + 1) \binom{i+j}{i-1}}\\
&= \frac{1}{(i + j + 1) \binom{i+j}{i}}
\end{align*}
which completes the proof.
\end{proof}

\begin{proof}[Proof of Theorem \ref{AS-SS-agree}]
By \ADD, it is enough to consider $f(r) = r_{i_1} r_{i_2}\cdots r_{i_k}$, since Lemma~\ref{addsepuni} shows that the two methods agree for additively separable functions.  Further, if $f(r)$ does not depend on the value of $r_i$, then the attribution to variable $r_i$ is $0$ by \DUM, so in fact it is enough to consider $f(r) = r_1 \cdot \cdots \cdot r_n$. 

In this case, recall that the Aumann-Shapley method is the affine path attribution method for $\gamma_i(t) = t$, so the attributions are given by 
\begin{align*}
z_i^{AS}(r, s) &= \int_0^1 \partial_if(\gamma_{r, s}(t)) \gamma_{r, s, i}'(t) dt \\
    &= \int_0^1 (s_i - r_i) \prod_{j \neq i} \left[r_j + (s_j - r_j) \gamma_j(t)\right] \gamma_i'(t) dt\\
    &=(s_i - r_i) \int_0^1 \gamma_i'(t) \sum_{K \subset [n] - \{i\}} \prod_{j \in K} r_j \prod_{j \in [n] - \{i\} - K} (s_j - r_j) \gamma_j(t) dt\\
    &=(s_i - r_i) \int_0^1 \gamma_i'(t)\!\!\!\! \sum_{J \subset [n]-\{i\}} \prod_{j \in J} s_j\!\!\!\! \prod_{j \in [n] - J - \{i\}}\!\! r_j\!\! \sum_{J \subset K \subset [n] - \{i\}} \!\!\!\! (-1)^{|K| - |J|} \prod_{j \in K} \gamma_j(t) dt \\
    &=(s_i - r_i) \sum_{J \subset [n] - \{i\}} \prod_{j \in J} s_j \prod_{j \in [n] - J - \{i\}} r_j \int_0^1 \gamma_i'(t) \prod_{j \in J} \gamma_j(t) \!\!\!\!\!\! \prod_{j \in [n] - J - \{i\}} \!\!\!\! (1 - \gamma_j(t)) dt,
\end{align*}
which is of the form
\[
z_i^{AS}(r, s) = (s_i - r_i) \sum_{J \subset [n] - \{i\}} c_{i, J} \prod_{j \in J} s_j \prod_{j \in [n] - J - \{i\}} r_j
\]
for the constants 
\[
c_{i, J} = \int_0^1 \gamma'_i(t) \prod_{j \in J} \gamma_j(t) \prod_{j \in [n] - J - \{i\}} (1 - \gamma_j(t)) dt = \int_0^1 t^{|J|} (1 - t)^{n - 1 - |J|} dt.
\]

On the other hand, each affine path attribution method for $\gamma^\sigma$ assigns to variable $i$ the attribution 
\[
z_i^{\sigma}(r, s) = (s_i - r_i) \prod_{\sigma(j) < \sigma(i)} r_j \prod_{\sigma(j) > \sigma(i)} s_j.
\]
Therefore, the attribution assigned to variable $i$ under Shapley-Shubik is 
\begin{align*}
z_i^{SS} = \frac{1}{n!} (s_i - r_i) \sum_{\sigma \in S_n}  \prod_{\sigma(j) < \sigma(i)} r_j \prod_{\sigma(j) > \sigma(i)} s_j = \frac{1}{n!} (s_i - r_i) \sum_{J \subset [n] - \{i\}} |J|! (n - 1 - |J|)! \prod_{j \in J} s_J \prod_{j \in [n] - \{i\} - J} r_j,
\end{align*}
so it suffices for us to show that 
\begin{equation*} 
\int_0^1 t^{|J|} (1 - t)^{n - 1 - |J|} dt = \frac{|J|! (n - 1 - |J|)!}{n!},
\end{equation*}
which follows by taking $i = |J|$ and $j = n - 1 - |J|$ in Lemma~\ref{binom-int}.
\end{proof}

\section{Affine path attribution methods} \label{affine-path-methods}


The Aumann-Shapley and the Shapley-Shubik methods are both affine path attribution methods. The following lemma demonstrates why they satisfy \ASI.

\begin{lemma} \label{affineasi}
Every affine path attribution method satisfies \ASI.
\end{lemma}
\begin{proof}
Let $z_i$ be the affine single-path attribution method corresponding to $\gamma$.  For any $c, d > 0$, set $g(r_1, \ldots, r_n) = f(r_1, \ldots, (r_j - d)/c, \ldots, r_n)$, $r' = (r_1, \ldots, cr_j + d, \ldots, r_n)$, and $s' = (s_1, \ldots, cs_j + d, \ldots, s_n)$.  Then, taking $\tau_{ij}(c) = c$ if $i = j$ and $\tau_{ij}(c) = 1$ otherwise, we have
\begin{align*}
z_i\Big(r', s', g\Big) &= \int_0^1 \partial_i g\Big(r' + \big((s_1 - r_1) \gamma_1(t), \ldots, c(s_j - r_j) \gamma_j(t), \ldots, (s_n - r_n) \gamma_n(t)\big)\Big) (s_i - r_i) \tau_{ij}(c) \gamma_i'(t)\, dt  \\
	&= \int_0^1 \frac{1}{\tau_{ij}(c)}\, \partial_i f(\gamma_{r, s}(t))\, (s_i - r_i)\, \tau_{ij}(c)\, \gamma_i'(t)\, dt\\
	&= \int_0^1 \partial_i f(\gamma_{r, s}(t)) \gamma_{r, s, i}'(t)\,dt\\
	&= z_i(r, s, f).
\end{align*}
The result follows because \ASI is preserved under convex combinations.
\end{proof}

\end{document}